\documentclass[copyright,creativecommons]{eptcs}
\providecommand{\event}{HPC'10}  
\providecommand{\firstpage}{1}
\usepackage{breakurl}
\usepackage{amsmath, amsfonts, amssymb, eucal, amsthm}

\input{Qcircuit}

\newcommand{\quotes}[1]{``#1''}
\newcommand{\comment}[1]{}
\newcommand{\braket}[2]{{\langle {#1}\!\mid\!{#2} \rangle}}
\newcommand{\Hilbert}{{\cal H}}
\newcommand{\GoodSetSize}[1]{2^{\lceil\log((2/\epsilon)\ln{2#1})\rceil}}
\newtheorem{definition}{Definition}
\newtheorem{lemma}{Lemma}
\newtheorem{corollary}{Corollary}
\newtheorem{theorem}{Theorem}

\begin{document}

\title{On Computational Power of Quantum\\ Read-Once Branching Programs}
\author{Farid Ablayev\thanks{Work was in part supported by the Russian Foundation for Basic Research
    under the grants 08-07-00449, 09-01-97004}
\institute{Kazan Federal University,\\ Institute for Informatics of Tatarstan Academy of Sciences\\
Kazan, Russian Federation} \email{fablayev@gmail.com}
\and Alexander Vasiliev
\thanks{Work was in part supported by the Russian Foundation for Basic Research
    under the grants 08-07-00449, 09-01-97004 and by Tatarstan Republic under the youth grant 01-5/2010(G).}
    \institute{Kazan Federal University,\\ Institute
for Informatics of Tatarstan Academy of Sciences\\ Kazan, Russian Federation}
\email{Alexander.KSU@gmail.com} }
\def\titlerunning{Quantum Branching Programs}
\def\authorrunning{F. Ablayev \& A. Vasiliev}

\maketitle

\begin{abstract}
In this paper we review our current results concerning the computational power of quantum read-once
branching programs. First of all, based on the circuit presentation of quantum branching programs
and our variant of quantum fingerprinting technique, we show that any Boolean function with linear
polynomial presentation can be computed by a quantum read-once branching program using a relatively
small (usually logarithmic in the size of input) number of qubits. Then we show that the described
class of Boolean functions is closed under the polynomial projections.
\end{abstract}


\section{Introduction}

One notable thing about the recent realizations of a quantum computer (say, the one based on
multiatomic ensembles in resonator \cite{Moiseev:2010:multi-ensembleQC}) is that they abide the
isolation of a quantum system, e.g. in \cite{Moiseev:2010:multi-ensembleQC} the transformation of a
quantum state is performed by an external magnetic field. Thus it is quite adequate to describe
such computations by quantum models with classical control. Here we consider one of such models --
the model of \emph{quantum branching programs} introduced by Ablayev, Gainutdinova, Karpinski
\cite{Ablayev-Gainutdinova-Karpinski:2001:QBP} (\emph{leveled programs}), and by Nakanishi,
Hamaguchi, Kashiwabara \cite{Nakanishi:2000:QBP} (\emph{non-leveled programs}). Later it was shown
by Sauerhoff \cite{Sauerhoff-Sieling:2005:QBP} that these two models are polynomially equivalent.

It is also worth noting that in spite of constant progress in experimental quantum computation all
of the physical implementations of a quantum computer are still rather weak in a sense that they
are suffering from fast decoherence of the quantum states and are able to organize the interaction of a
small number of qubits. This naturally leads to the restricted variants of a quantum computer --
the idea first proposed by Ambainis and Freivalds in 1998 \cite{Ambainis-Freivalds:1998:QFA}.
Considering one-way quantum finite automata, they suggested that the first quantum-mechanical computers
would consist of a comparatively simple and fast quantum-mechanical part connected to a classical
computer.

In this paper we consider a restricted model of computation known as \emph{Quantum Read-Once
Branching Programs} of polynomial width. The classical variant of this model is also known in
computer science as Ordered Binary Decision Diagrams (OBDDs) and that is why we will also use the
notion of quantum OBDDs (QOBDDs) for the considered model. The small coherence time is formalized
in this model by allowing only a single test of each variable. The restriction of polynomial width
of quantum OBDDs leads to the number of qubits which is logarithmic in the size of the input. On
the other hand, the generalized lower bound on the width of quantum OBDDs
\cite{Ablayev-et-al:2005:QBP} leads to logarithmic number of qubits as a lower bound for almost all
Boolean functions in the OBDD setting.

For the model of quantum OBDDs we develop the \emph{fingerprinting} technique introduced in
\cite{ablayev-vasiliev:2008:ECCC}. The basic ideas of this approach are due to Freivalds (e.g. see
the book \cite{Motwani:1995:Randomized-algorithms}). These ideas were later successfully applied in
the \emph{quantum automata} setting by Ambainis and Freivalds in 1998
\cite{Ambainis-Freivalds:1998:QFA} (later improved in \cite{Ambainis-Nahimovs:2008:QFA}).
Subsequently, the same technique was adapted for the quantum branching programs by Ablayev,
Gainutdinova and Karpinski in 2001 \cite{Ablayev-Gainutdinova-Karpinski:2001:QBP}, and was later
generalized in \cite{ablayev-vasiliev:2008:ECCC}.

For our technique we use a variant of polynomial presentation for
Boolean functions, which we call \emph{characteristic}. The
polynomial presentations of Boolean functions are widely used in
theoretical computer science. For instance, an algebraic
transformation of Boolean functions has been applied in
\cite{Jain:1992:verification} and
\cite{Agrawal:1998:characteristic-polynomials} for verification of
Boolean functions. In the quantum setting polynomial representations
were used for proving lower bounds on communication complexity in
\cite{Buhrman:2001:Fingerprinting} as well as for investigating query
complexity in \cite{Wolf:2001:PhD}. Our approach combines the ideas
similar to the definition of characteristic polynomial from
\cite{Jain:1992:verification},
\cite{Agrawal:1998:characteristic-polynomials} and to the notion of
\emph{zero-error polynomial} (see, e.g. \cite{Wolf:2001:PhD}).

Finally, we use the technique of polynomial projections to outline
the limits of our fingerprinting method. This technique was
intensively applied by Sauerhoff in the model of classical branching
programs (see, e.g., \cite{Sauerhoff:2001:randomized-BPs}). In this
paper we apply this approach for the quantum OBDD model by showing
that the class of functions effectively computable via the quantum
fingerprinting technique is closed under polynomial projections.

\section{Preliminaries}

We use the notation $\ket{i}$ for the vector from $\Hilbert^d$,
which has a $1$ on the $i$-th position and $0$ elsewhere. An orthonormal
basis $\ket{1}$,\ldots,$\ket{d}$ is usually referred to as the
\emph{standard computational basis}. In this paper we consider all
quantum transformations and measurements with respect to this basis.

\begin{definition}
A Quantum Branching Program ${Q}$ over the Hilbert space
$\Hilbert^d$ is defined as
\[
Q=\langle T, \ket{\psi_0}, Accept\rangle,
\]
where $T$ is a sequence of $l$ instructions: $T_j=\left(x_{i_j},
U_j(0),U_j(1)\right)$ is determined by the variable $x_{i_j}$ tested
on the step $j$, and $U_j(0)$, $U_j(1)$ are unitary transformations
in $\Hilbert^d$.

Vectors $\ket{\psi}\in \Hilbert^d$ are called states (state vectors)
of $Q$, $\ket{\psi_0}\in \Hilbert^d$ is the initial state of $Q$,
and $Accept\subseteq\{1,2,\ldots d\}$ is the set of indices of accepting basis states.

We define a computation of ${Q}$ on an input $\sigma = \sigma_1
\ldots \sigma_n \in \{0,1\}^n$  as follows:
\begin{enumerate}
\item A computation of ${Q}$ starts from the initial state
      $\ket{\psi_0}$;
\item The $j$-th instruction of $Q$ reads the input symbol $\sigma_{i_j}$ (the
value of $x_{i_j}$) and applies the transition matrix $U_j =
      U_j(\sigma_{i_j})$ to the current state $\ket{\psi}$ to
      obtain the state
      $\ket{\psi'}=U_j(\sigma_{i_j})\ket{\psi}$;
\item The final state is
\[ \ket{\psi_\sigma}= \left(\prod_{j=l}^1 U_j(\sigma_{i_j})\right)
\ket{\psi_0}\enspace . \]
\item After the $l$-th (last) step of quantum transformation $Q$ measures
its configuration $\ket{\psi_\sigma}$ =\linebreak
$(\alpha_1,\ldots, \alpha_d)^T$, and the input $\sigma$ is
accepted with probability
\[Pr_{accept}(\sigma)=\sum\limits_{i\in
Accept}|\alpha_i|^2.\]
\end{enumerate}

\end{definition}

Note, that using the set $Accept$ we can construct $M_{accept}$ -- a
projector on the accepting subspace $\Hilbert^{d}_{accept}$ (i.e. a
diagonal zero-one projection matrix, which determines the final
projective measurement). Thus, the accepting probability can be
re-written as
\[ Pr_{accept}(\sigma)=\braket{\psi_\sigma M^\dag_{accept}}{M_{accept}\psi_\sigma}=
||M_{accept}\ket{\psi_\sigma}||^2_2. \]

Note also that this is a ``measure-once'' model analogous to the
model of quantum finite automata in~\cite{Moore:2000:QFA}, in which
the system evolves unitarily except for a single measurement at the
end.  We could also allow multiple measurements during the
computation, by representing the state as a density matrix $\rho$,
and by making the $U_j$ superoperators, but we do not consider this
here.

\paragraph{Circuit representation.}

Quantum algorithms are usually given by using quantum circuit
formalism \cite{Deutsch:1989:QuantumCircuits, Yao:1993:QCircuits},
because this approach is quite straightforward for describing such
algorithms.

We propose, that a QBP represents a classically-controlled quantum
system. That is, a QBP can be viewed as a quantum circuit aided with
an ability to read classical bits as control variables for unitary
operations. 
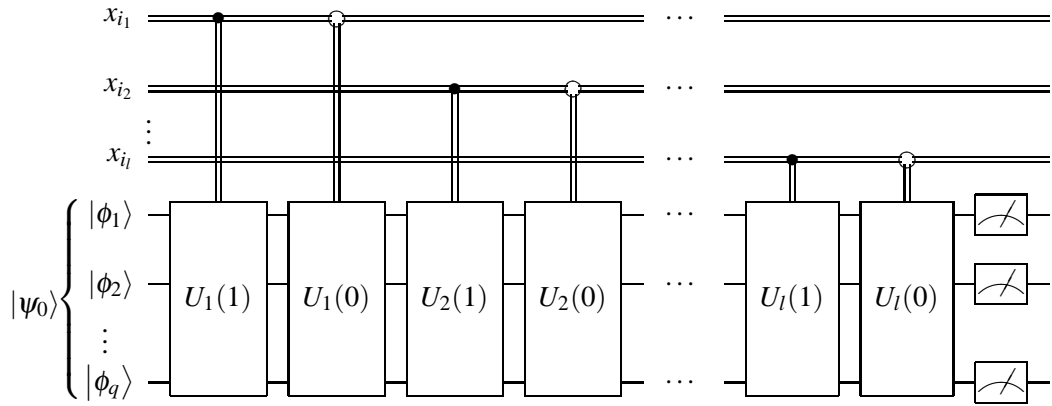
\[\quad\quad
\Qcircuit  @C=0.75em @R=1.0em {
&&&\lstick{x_{i_1}} & \control\cw\cwx[5] & \controlo\cw\cwx[5] & \cw & \cw & \cw & \push{\cdots\quad} & \cw & \cw & \cw & \cw\\
\\
&&&\lstick{x_{i_2}} & \cw & \cw & \control\cw\cwx[3] & \controlo\cw\cwx[3] & \cw & \push{\cdots\quad} & \cw & \cw & \cw & \cw\\
&&&\vdots \\
&&&\lstick{x_{i_l}} & \cw & \cw & \cw & \cw & \cw  & \push{\cdots\quad} & \control\cw\cwx[1] & \controlo\cw\cwx[1] & \cw & \cw\\
&&&\lstick{\ket{\phi_1}} & \multigate{3}{U_1(1)} & \multigate{3}{U_1(0)} & \multigate{3}{U_2(1)} & \multigate{3}{U_2(0)} & \qw  & \push{\cdots\quad} & \multigate{3}{U_l(1)} & \multigate{3}{U_l(0)} & \meter & \qw\\
&&&\lstick{\ket{\phi_2}} & \ghost{U_1(1)} & \ghost{U_1(0)} & \ghost{U_2(1)} & \ghost{U_2(0)} & \qw & \push{\cdots\quad} & \ghost{U_l(1)} & \ghost{U_l(0)} & \meter & \qw\\
\ustick{\ket{\psi_0}\quad\quad\quad~} \gategroup{6}{1}{9}{1}{1em}{\{} & \vdots \\
&&&\lstick{\ket{\phi_q}} & \ghost{U_1(1)} & \ghost{U_1(0)} & \ghost{U_2(1)} & \ghost{U_2(0)} & \qw & \push{\cdots\quad} & \ghost{U_l(1)} & \ghost{U_l(0)} & \meter & \qw\\
}
\]
Here $x_{i_1},\ldots,x_{i_l}$ is the sequence of (not necessarily
distinct) variables denoting classical control bits. Using the
common notation single wires carry quantum information and double
wires denote classical information and control.

\paragraph{Complexity measures.}
The width of a QBP $Q$, denoted by $width(Q)$, is the dimension $d$
 of the corresponding state space
$\Hilbert^d$, and the length of $Q$, denoted by $length(Q)$, is the number $l$ of
instructions in the sequence $T$. There is one more commonly used complexity measure --
the size of Q, which we define as $size(Q)=width(Q)\cdot length(Q)$.

Note that for a QBP $Q$ in the circuit setting another important
complexity measure explicitly comes out -- a number of quantum bits,
denoted by $qubits(Q)$,
physically needed to implement a corresponding quantum system with
classical control. From definition it follows that $\log width(Q)\leq qubits(Q)$.


\paragraph{Acceptance criteria.} A $\mathrm{QBP}$ $Q$ \emph{computes the Boolean function $f$ with
bounded error} if there exists an $\epsilon\in(0,1/2)$ (called
\emph{margin}) such that for all inputs the probability of error is
bounded by $1/2-\epsilon$.\\

In particular, we say that a $\mathrm{QBP}$ $Q$ \emph{computes the
Boolean function $f$ with one-sided error} if there exists an
$\epsilon\in(0,1)$ (called \emph{error}) such that for all
$\sigma\in f^{-1}(1)$ the probability of $Q$ accepting $\sigma$ is 1
and for all $\sigma\in f^{-1}(0)$ the probability of $Q$ erroneously
accepting $\sigma$ is less than $\epsilon$.

\paragraph{Read-once branching programs.}

Read-once BPs is a well-known restricted variant of branching
programs \cite{Wegener:2000:BP}.

\begin{definition}
We call a $\mathrm{QBP}$ $Q$ a quantum $\mathrm{OBDD}$
($\mathrm{QOBDD}$) or read-once $\mathrm{QBP}$ if each variable
$x\in\{x_1,\dots,x_n\}$ occurs in the sequence $T$ of
transformations of $Q$ at most once.
\end{definition}

For the rest of the paper we're only interested in QOBDDs, i.e. the
length of all programs would be $n$ (the number of input variables).
Note that for OBDD model $size(Q)=n\cdot width(Q)$ and therefore
we're mostly interested in the width of quantum OBDDs.

\paragraph{Generalized Lower Bound.}
The following general lower bound on the width of QOBDDs was proven
in \cite{Ablayev-et-al:2005:QBP}.

\begin{theorem}\label{General-Lower-Bound}
 Let $f(x_1, \ldots, x_n)$ be a Boolean function
 computed by a quantum read-once branching program $Q$ with bounded error for some
 margin $\epsilon$. Then
 \[ {\rm width}(Q)\geq\frac{\log {\rm width}(P)}{2\log\left(1+\frac{1}{\epsilon}\right)} \]
 where $P$  is a deterministic OBDD of minimal width computing
 $f(x_1, \ldots, x_n)$.
\end{theorem}

That is, the width of a quantum OBDD cannot be asymptotically less
than the logarithm of the width of the minimal deterministic OBDD
computing the same function. And since the deterministic width of
many \quotes{natural} functions is exponential
\cite{Wegener:2000:BP}, we obtain the linear lower bound for these
functions.

Let $bits(P)$  be  the number of bits (memory size) required to
implement the minimal deterministic OBDD  $P$ for $f$ and $Q$ is an
arbitrary quantum OBDD computing the same function.

Then  theorem \ref{General-Lower-Bound} can be restated as the
following corollary using the number of bits and qubits as the
complexity measure.
\begin{corollary}
$$qubits(Q)=\Omega(\log{bits(P)}).$$
\end{corollary}

\section{Algorithms for QBPs Based on Fingerprinting}

Generally \cite{Motwani:1995:Randomized-algorithms},
\emph{fingerprinting} -- is a technique that allows to present
objects (words over some finite alphabet) by their
\emph{fingerprints}, which are significantly smaller than the
originals. It is used in randomized and quantum algorithms to test
\emph{equality} of some objects (binary strings) with one-sided
error by simply comparing their fingerprints.

In this paper we develop a variant of the fingerprinting technique
adapted for quantum branching programs. At the heart of the method
is the representation of Boolean functions by polynomials of special
type, which we call \emph{characteristic}.

\subsection{Characteristic Polynomials for Quantum Fingerprinting}

The following definition is similar to the algebraic transformation
of Boolean function from \cite{Jain:1992:verification} and
\cite{Agrawal:1998:characteristic-polynomials}, though it is adapted
for the fingerprinting technique.

\begin{definition}\label{polynomial-definition}
We call a polynomial $g(x_1,\dots,x_n)$ over the ring ${\mathbb
Z}_m$ a characteristic polynomial of a Boolean function
$f(x_1,\ldots,x_n)$ and denote it $g_f$
%
when for all $\sigma\in\{0,1\}^n$ 
$g_f(\sigma)=0$ iff $f(\sigma)=1$.
\end{definition}

%
\begin{lemma}\label{Existence-Of-A-Characteristic-Polynomial}
For any Boolean function $f$ there exists a characteristic
polynomial $g_f$ over ${\mathbb Z}_{2^n}$.
\end{lemma}
\begin{proof}
One way to construct such characteristic polynomial $g_f$ is
transforming a sum of products representation for $\neg f$.

Let $K_1\vee\ldots\vee K_l$ be a sum of products for $\neg f$ and
let $\tilde{K_i}$ be a product of terms from $K_i$ (negations $\neg
x_j$ are replaced by $1-x_j$). Then $\tilde{K_1}+\ldots+\tilde{K_l}$
is a characteristic polynomial over ${\mathbb Z}_{2^n}$ for $f$
since it equals $0$ $\iff$ all of $\tilde{K_i}$ (and thus $K_i$)
equal $0$. This happens only when the negation of $f$ equals $0$.
\end{proof}

Generally, there are many polynomials for the same function, but
we're interested in having a linear polynomial if it exists.

For example, the function $EQ_n$, which tests the equality of two
$n$-bit binary strings, has the following polynomial over ${\mathbb
Z}_{2^n}$:
$$\sum\limits_{i=1}^n\left(x_i(1-y_i)+(1-x_i)y_i\right)=\sum\limits_{i=1}^n\left(x_i+y_i-2x_iy_i\right).$$

On the other hand, the same function can be represented by the
polynomial
$$\sum\limits_{i=1}^nx_i2^{i-1}-\sum\limits_{i=1}^ny_i2^{i-1}.$$

Some functions don't have a linear characteristic polynomial at all
(e.g. a disjunction of $n$ variables $f=x_1\vee x_2\vee \ldots \vee
x_n$), while their negations have a linear characteristic polynomial
(e.g. $g_{\neg f}= \sum_{i=1}^n x_i$ over ${\mathbb Z}_{n+1}$).

We use this presentation of Boolean functions for our fingerprinting
technique.

\subsection{Fingerprinting technique}\label{fqbp}

For a Boolean function $f$ we choose an error rate $\epsilon>0$ and
pick a characteristic polynomial $g$ over the ring $\mathbb{Z}_m$.
Then for arbitrary binary string $\sigma=\sigma_1\ldots \sigma_n$ we
create its fingerprint $\ket{h_\sigma}$ composing
$t=\GoodSetSize{m}$ single qubit fingerprints $\ket{h^i_\sigma}$:
$$\ket{h^i_\sigma}  = \cos\frac{2\pi k_i
g(\sigma)}{m}\ket{0}+\sin\frac{2\pi k_i g(\sigma)}{m}\ket{1}$$
into entangled state $\ket{h_\sigma}$ of $\log{t}+1$ qubits:
$$\ket{h_\sigma} =
\frac{1}{\sqrt{t}}\sum\limits_{i=1}^t\ket{i}\ket{h^i_\sigma}.
$$
Here the transformations of the last qubit in $t$ different subspaces \quotes{simulate}
the transformations of all of the $\ket{h^i_\sigma}$ ($i=1,\ldots,t$).
%
%
That is, the last qubit is in parallel rotated by $t$ different angles about the
$\hat{y}$ axis of the Bloch sphere.

The chosen parameters $k_i\in\{1,\ldots,m-1\}$ for
$i\in\{1,\ldots,t\}$ are ``good'' following the notion of
\cite{Ambainis-Freivalds:1998:QFA}.

\begin{definition}\label{good-set}
A set of parameters $K=\{k_1, \dots, k_t\}$ is called ``good'' for
some integer $b\neq0 \bmod m$ if
$$\frac{1}{t^2}\left(\sum\limits_{i=1}^t\cos{\frac{2\pi k_i b}{m}}\right)^2<\epsilon.$$
\end{definition}
The left side of inequality is the squared amplitude of the basis
state $\ket{0}^{\otimes\log{t}}\ket{0}$ if $b=g(\sigma)$ and the
operator $H^{\otimes\log{t}}\otimes I$ has been applied to the
fingerprint $\ket{h_\sigma}$. Informally, that kind of set
guarantees, that the probability of error will be bounded by a
constant below 1.

The following lemma proves the existence of a ``good'' set and
generalizes the proof of the corresponding statement from
\cite{Ambainis-Nahimovs:2008:QFA}.

\begin{lemma}\cite{ablayev-vasiliev:2008:ECCC}\label{existence-of-a-good-set}
There is a set $K$ with $|K|=t=\GoodSetSize{m}$ which is ``good''
for all integer $b\neq0 \bmod m$.
\end{lemma}

\comment{

The proof of the previous lemma doesn't give an algorithm for the
construction of a ``good'' set. But in
\cite{Ambainis-Nahimovs:2008:QFA} there is an explicit approach
based on ideas from \cite{aikps90}, which can be used to construct a
bigger, but still ``good'' set.

Fix $\epsilon>0$ and introduce the following notation:
$$P = \{p|\ p\mbox{ is prime and } (\log{m})^{1+\epsilon}/2<p\leq
(\log{m})^{1+\epsilon}\},$$
$$S = \{1, 2, \ldots, (\log{m})^{1+2\epsilon}\},$$
$$K = \{s\cdot p^{-1}|\ s\in S, p\in P\},$$
where $p^{-1}$ is the inverse modulo $m$.

It is obvious that $|K| = O(\log^{2+3\epsilon}m)$. It turns out
\cite{aikps90} that for each $g\in\{1, 2$, \ldots, $m-1\}$
$$\frac{1}{|K|}\left|\sum\limits_{j=1}^t e^{\frac{2\pi k_j g}{m}i}\right|\leq(\log{m})^{-\epsilon}.$$

Taking the real part of the previous inequality we induce that
$$\frac{1}{t}\left|\sum\limits_{j=1}^t \cos{\frac{2\pi k_j g}{m}}\right|\leq(\log{m})^{-\epsilon}$$
for each $g\in\{1, 2, \ldots, m-1\}$.

Thus, the set $K$ is ``good'' for all $g\neq 0 \bmod{m}$.

}

We use this result for our fingerprinting technique choosing  the
set $K=\{k_{1}, \dots, k_{t}\}$ which is ``good'' for all
$b=g(\sigma)\neq0$. That is, it allows to distinguish those inputs
whose image is 0 modulo $m$ from the others.


\subsection{Boolean Functions Computable via Fingerprinting Method}
\comment{ It is obvious that the complexity of computing some
Boolean function $f$ via the fingerprinting technique propagates to
the creating of its fingerprint. Now the question is what types of
the mapping $g$ allow effective computation of the corresponding
fingerprint.

Generally, there are many polynomials for the same function. But not
all of them are suitable for our method. }

Let $f(x_1,\ldots, x_n)$ be a Boolean function and $g(x_1,\ldots, x_n)$ be its
characteristic polynomial. The following theorem holds.

\begin{theorem}\label{LinearPolynomialComputation}
Let $\epsilon\in(0,1)$. If $g$ is a linear polynomial over
$\mathbb{Z}_m$, i.e. $g=c_1x_1+\ldots c_nx_n+c_0$, then $f$ can be
computed with one-sided error $\epsilon$ by a quantum OBDD of width
$O\left(\frac{\log m}{\epsilon}\right)$.
\end{theorem}
\begin{proof}
Here is the algorithm in the circuit notation:
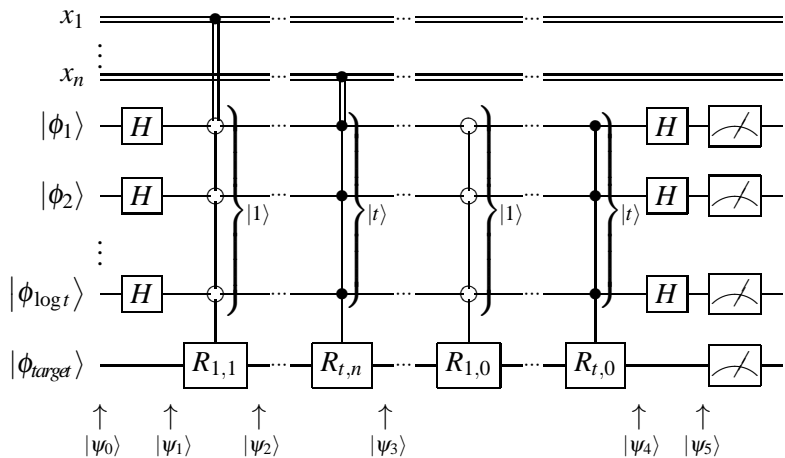
\[\quad\quad\quad
\Qcircuit @C=0.75em @R=1.0em {
\lstick{x_1} & \cw & \control\cw\cwx[3] & \cw & ~_{\cdots}\quad & \cw & \cw & ~_{\cdots}\quad & \cw & \cw & ~_{\cdots}\quad & \cw & \cw & \cw & \cw\\
\vdots \\
\lstick{x_n} & \cw & \cw & \cw & ~_{\cdots}\quad & \control\cw\cwx[1] & \cw & ~_{\cdots}\quad & \cw & \cw & ~_{\cdots}\quad & \cw & \cw & \cw & \cw\\
\lstick{\ket{\phi_1}} & \gate{H} & \ctrlo{1} & \qw  & ~_{\cdots}\quad & \ctrl{1} & \qw & ~_{\cdots}\quad & \ctrlo{1} & \qw  & ~_{\cdots}\quad & \ctrl{1} & \gate{H} & \meter & \qw\\
\lstick{\ket{\phi_2}} & \gate{H} & \ctrlo{2} & \qw  & ~_{\cdots}\quad & \ctrl{2} &  \qw & ~_{\cdots}\quad & \ctrlo{2} & \qw  & ~_{\cdots}\quad & \ctrl{2} & \gate{H} & \meter & \qw\\
\vdots &&\ustick{\quad\quad\quad^{\ket{1}}} &&& \ustick{\quad\quad~~^{\ket{t}}}&&& \ustick{\quad\quad\quad^{\ket{1}}}&&& \ustick{\quad\quad~~^{\ket{t}}}\gategroup{4}{3}{7}{3}{1em}{\}}\gategroup{4}{6}{7}{6}{1em}{\}}\gategroup{4}{9}{7}{9}{1em}{\}}\gategroup{4}{12}{7}{12}{1em}{\}} \\
\lstick{\ket{\phi_{\, \log t}}} & \gate{H} & \ctrlo{1} & \qw & ~_{\cdots}\quad & \ctrl{1}& \qw & ~_{\cdots}\quad & \ctrlo{1} & \qw  & ~_{\cdots}\quad & \ctrl{1} & \gate{H} & \meter & \qw\\
\lstick{\ket{\phi_{t\!a\!r\!g\!e\!t}}}& \qw & \gate{R_{1,1}} & \qw & ~_{\cdots}\quad & \gate{R_{t,n}} & \qw & ~_{\cdots}\quad & \gate{R_{1,0}} & \qw & ~_{\cdots}\quad & \gate{R_{t,0}} & \qw & \meter & \qw\\
\uparrow & \quad\quad\uparrow &\quad\quad\quad\uparrow&&& \quad\quad\quad\uparrow &&&&& & \quad\quad\quad\uparrow & \quad\quad~\uparrow\\
~_{\ket{\psi_0}} & \quad\quad~_{\ket{\psi_1}} &\quad\quad\quad~_{\ket{\psi_2}}&&& \quad\quad\quad~_{\ket{\psi_3}} &&&&& & \quad\quad\quad~_{\ket{\psi_4}} & \quad\quad~~_{\ket{\psi_5}}\\
}
\]

Initially qubits $\ket{\phi_1}\otimes\ket{\phi_2} \otimes\dots \otimes\ket{\phi_{\log
t}} \otimes \ket{\phi_{t\!a\!r\!g\!e\!t}}$ are in the state
$\ket{\psi_0}=\ket{0}^{\otimes\log{t}}\ket{0}$. For $i\in\{1,\dots, t\}$,
$j\in\{0,\dots, n\}$ we define rotations $R_{i,j}$ as
$$R_{i,j}=R_{\hat{y}}\left(\frac{4\pi k_i c_j}{m}\right),$$
where $c_j$ are the coefficients of the linear polynomial for $f$ and the set of
parameters $K=\{k_{1}, \dots, k_{t}\}$ is ``good'' according to the Definition
\ref{good-set} with $t=\GoodSetSize{\cdot m}$.

Let $\sigma=\sigma_1\ldots\sigma_{n} \in \{0,1\}^{n}$ be an input string.

The first layer of Hadamard operators transforms the state $\ket{\psi_0}$ into
$$\ket{\psi_1} = \frac{1}{\sqrt{t}}\sum\limits_{i=1}^t\ket{i}\ket{0}.$$

Next, upon input symbol $0$ identity transformation $I$ is applied. But if the value of
$x_j$ is $1$, then the state of the last qubit is transformed by the operator $R_{i,j}$,
rotating it by the angle proportional to $c_j$. Moreover, the rotation is done in each
of $t$ subspaces with the corresponding amplitude $1/\sqrt{t}$. Such a parallelism is
implemented by the controlled operators $C_i(R_{i,j})$, which transform the states
$\ket{i}\ket{\cdot}$ into $\ket{i}R_{i,j}\ket{\cdot}$, and leave others unchanged. For
instance, having read the input symbol $x_1=1$, the system would evolve into state
$$\begin{array}{rcl}
\ket{\psi_2} & = & \frac{1}{\sqrt{t}}\sum\limits_{i=1}^tC_i(R_{i,1})\ket{i}\ket{0} =
\frac{1}{\sqrt{t}}\sum\limits_{i=1}^t\ket{i}R_{i,1}\ket{0}\\
& = & \frac{1}{\sqrt{t}}\sum\limits_{i=1}^t\ket{i}\left(\cos\frac{2\pi k_i
c_1}{m}\ket{0} + \sin\frac{2\pi k_i c_1}{m}\ket{1}\right)
\end{array}.$$

Thus, after having read the input $\sigma$ the amplitudes would ``collect'' the sum
$\sum_{j=1}^nc_j\sigma_j$
$$
\begin{array}{rcl}
\ket{\psi_3} & = & \frac{1}{\sqrt{t}}\sum\limits_{i=1}^t\ket{i}\left(\cos\frac{2\pi k_i
\sum_{j=1}^nc_j\sigma_j}{m}\ket{0} + \sin\frac{2\pi k_i
\sum_{j=1}^nc_j\sigma_j}{m}\ket{1}\right)
\end{array}.
$$

At the next step we perform the rotations by the angle $\frac{4\pi k_i c_0}{m}$ about
the $\hat{y}$ axis of the Bloch sphere for each $i\in\{1,\dots,t\}$. Therefore, the
state of the system would be
$$\begin{array}{rcl} \ket{\psi_4} & = &
\frac{1}{\sqrt{t}}\sum\limits_{i=1}^t\ket{i}\left(\cos\frac{2\pi k_i
g(\sigma)}{m}\ket{0}+\sin\frac{2\pi k_i g(\sigma)}{m}\ket{1}\right).
\end{array}$$

Applying $H^{\otimes\log{t}}\otimes I$ we obtain the state
$$\begin{array}{rcl}
\ket{\psi_5} & = & \left(\frac{1}{t}\sum\limits_{i=1}^t\cos\frac{2\pi k_i
g(\sigma)}{m}\right)\ket{0}^{\otimes\log{t}}\ket{0}+\\
&& + \gamma\ket{0}^{\otimes\log{t}}\ket{1} +
\sum\limits_{i=2}^{t}\ket{i}\left(\alpha_i\ket{0} + \beta_i\ket{1}\right),
\end{array}$$
where $\gamma$, $\alpha_i$, and $\beta_i$ are some unimportant amplitudes.

The input $\sigma$ is accepted if the measurement outcome is
$\ket{0}^{\otimes\log{t}}\ket{0}$. Clearly, the accepting probability is
\[ Pr_{accept}(\sigma) = \frac{1}{t^2}\left(\sum\limits_{i=1}^t\cos\frac{
  2\pi k_i g(\sigma)}{2^{n}}\right)^2.
\]

If $f(\sigma)=1$ then $g(\sigma)=0$ and the program accepts $\sigma$ with probability
$1$. Otherwise, the choice of the set $K=\{k_1,\dots,k_t\}$ guarantees that
$$Pr_{accept}(\sigma) = \frac{1}{t^2}\left(\sum\limits_{i=1}^t\cos\frac{
  2\pi k_i g(\sigma)}{2^{n}}\right)^2<\epsilon.$$

Thus, $f$ can be computed by a quantum OBDD $Q$, with $qubits(Q)=\log{t}+1=O(\log\left(\frac{\log m}{\varepsilon}\right))$.
The width of the program is $2^{qubits(Q)}=O\left(\frac{\log m}{\varepsilon}\right)$.
\end{proof}

The following functions (for definitions see, e.g.,
\cite{ablayev-vasiliev:2009:EPTCS}) have the aforementioned linear
polynomials and thus are effectively computed via the fingerprinting
technique.

$$\begin{array}{llll}
  MOD_m & \sum\limits_{i=1}^nx_i & \mathbb{Z}_{m} & O(\log\log{m})\\
  \hline
  MOD'_m & \sum\limits_{i=1}^nx_i2^{i-1} & \mathbb{Z}_{m} & O(\log\log{m}) \\
  \hline
  EQ_n & \sum\limits_{i=1}^nx_i2^{i-1}- \sum\limits_{i=1}^ny_i2^{i-1} & \mathbb{Z}_{2^n} & O(\log{n}) \\
  \hline
  Palindrome_n & \sum\limits_{i=1}^{\lfloor n/2\rfloor}x_i2^{i-1}- \sum\limits_{i=\lceil n/2\rceil}^nx_i2^{n-i} & \mathbb{Z}_{2^{\lfloor n/2\rfloor}} & O(\log{n}) \\
  \hline
  PERM_n &
  \sum\limits_{i=1}^n\sum\limits_{j=1}^nx_{ij}\left((n+1)^{i-1}+(n+1)^{n+j-1}\right) & \mathbb{Z}_{(n+1)^{2n}} & O(\log{n})\\
          &  - \sum\limits_{i=1}^{2n}(n+1)^{i-1} &&\\
\end{array}$$

%

\section{Generalized Approach}

The fingerprinting technique described in the previous section
allows us to test a single property of the input encoded by a
characteristic polynomial. Using the same ideas we can test the
conjunction of several conditions encoded by a group of
characteristic polynomials which we call a \emph{characteristic} of
a function.

\begin{definition}
We call a set $\chi_f$ of polynomials over $\mathbb{Z}_m$ a
\emph{characteristic} of a Boolean function $f$ if for all
polynomials $g\in\chi_f$ and all $\sigma\in\{0,1\}^n$ it holds
that $g(\sigma)=0$ iff $\sigma\in f^{-1}(1)$.
\end{definition}

We say that a characteristic is \emph{linear} if all of its
polynomials are linear.

From Lemma \ref{Existence-Of-A-Characteristic-Polynomial} it follows
that for each Boolean function there is always a characteristic
consisting of a single characteristic polynomial.

Now we can generalize the Fingerprinting technique from section
\ref{fqbp}.

\paragraph{Generalized Fingerprinting technique} For a Boolean function $f$ we
choose an error rate $\epsilon>0$ and pick a characteristic
$\chi_f=\{g_1,\ldots, g_l\}$ over $\mathbb{Z}_m$. Then for arbitrary binary string
$\sigma=\sigma_1\ldots \sigma_n$ we create its fingerprint
$\ket{h_\sigma}$ composing $t\cdot l$ ($t=\GoodSetSize{m}$) single
qubit fingerprints $\ket{h^i_\sigma(j)}$:
$$\begin{array}{rcl}
\ket{h^i_\sigma(j)} & = & \cos\frac{\pi k_i
g_j(\sigma)}{m}\ket{0}+\sin\frac{\pi k_i g_j(\sigma)}{m}\ket{1}\\
\ket{h_\sigma} & = &
\frac{1}{\sqrt{t}}\sum\limits_{i=1}^t\ket{i}\ket{h^i_\sigma(1)}\ket{h^i_\sigma(2)}\ldots\ket{h^i_\sigma(l)}
\end{array}$$

\begin{theorem}\label{LinearCharacteristicComputation}
If $\chi_f$ is a linear characteristic then $f$ can be computed by
a quantum OBDD of width $O(2^{|\chi_f|}\log m)$.
\end{theorem}
\begin{proof}
The proof of this result somewhat generalizes the
proof of Theorem \ref{LinearPolynomialComputation} to the case of
several target qubits. Here is the circuit:
$$\quad\quad\quad
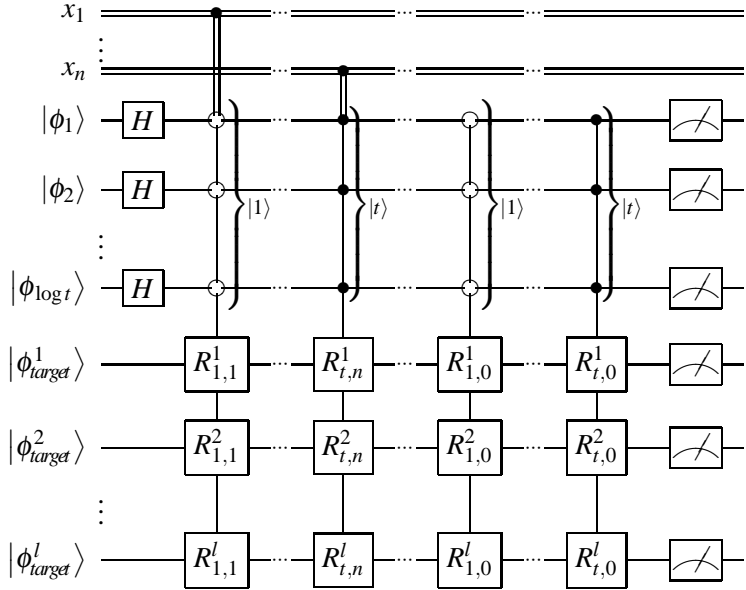
\Qcircuit @C=0.75em @R=1.0em {
\lstick{x_1} & \cw & \control\cw\cwx[3] & \cw & ~_{\cdots}\quad & \cw & \cw & ~_{\cdots}\quad & \cw & \cw & ~_{\cdots}\quad & \cw & \cw & \cw & \cw\\
\vdots \\
\lstick{x_n} & \cw & \cw & \cw & ~_{\cdots}\quad & \control\cw\cwx[1] & \cw & ~_{\cdots}\quad & \cw & \cw & ~_{\cdots}\quad & \cw & \cw & \cw & \cw\\
\lstick{\ket{\phi_1}} & \gate{H} & \ctrlo{1} & \qw  & ~_{\cdots}\quad & \ctrl{1} & \qw & ~_{\cdots}\quad & \ctrlo{1} & \qw  & ~_{\cdots}\quad & \ctrl{1} & \qw & \meter & \qw\\
\lstick{\ket{\phi_2}} & \gate{H} & \ctrlo{2} & \qw  & ~_{\cdots}\quad & \ctrl{2} &  \qw & ~_{\cdots}\quad & \ctrlo{2} & \qw  & ~_{\cdots}\quad & \ctrl{2} & \qw & \meter & \qw\\
\vdots &&\ustick{\quad\quad\quad^{\ket{1}}} &&& \ustick{\quad\quad~~^{\ket{t}}}&&& \ustick{\quad\quad\quad^{\ket{1}}}&&& \ustick{\quad\quad~~^{\ket{t}}}\gategroup{4}{3}{7}{3}{1em}{\}}\gategroup{4}{6}{7}{6}{1em}{\}}\gategroup{4}{9}{7}{9}{1em}{\}}\gategroup{4}{12}{7}{12}{1em}{\}} \\
\lstick{\ket{\phi_{\, \log t}}} & \gate{H} & \ctrlo{1} & \qw & ~_{\cdots}\quad & \ctrl{1}& \qw & ~_{\cdots}\quad & \ctrlo{1} & \qw  & ~_{\cdots}\quad & \ctrl{1} & \qw & \meter & \qw\\
\lstick{\ket{\phi^1_{t\!a\!r\!g\!e\!t}}} & \qw & \gate{R^1_{1,1}}\qwx[1] & \qw & ~_{\cdots}\quad & \gate{R^1_{t,n}}\qwx[1] & \qw & ~_{\cdots}\quad & \gate{R^1_{1,0}}\qwx[1] & \qw & ~_{\cdots}\quad & \gate{R^1_{t,0}}\qwx[1] & \qw & \meter & \qw\\
\lstick{\ket{\phi^2_{t\!a\!r\!g\!e\!t}}} & \qw & \gate{R^2_{1,1}}\qwx[2] & \qw & ~_{\cdots}\quad & \gate{R^2_{t,n}}\qwx[2] & \qw & ~_{\cdots}\quad & \gate{R^2_{1,0}}\qwx[2] & \qw & ~_{\cdots}\quad & \gate{R^2_{t,0}}\qwx[2] & \qw & \meter & \qw\\
\vdots \\
\lstick{\ket{\phi^l_{t\!a\!r\!g\!e\!t}}} & \qw & \gate{R^l_{1,1}} & \qw & ~_{\cdots}\quad & \gate{R^l_{t,n}} & \qw & ~_{\cdots}\quad & \gate{R^l_{1,0}} & \qw & ~_{\cdots}\quad & \gate{R^l_{t,0}} & \qw & \meter & \qw\\
}
$$
\begin{enumerate}
\item Upon the input $\sigma=\sigma_1\ldots\sigma_n$ we create the
fingerprint $\ket{h_\sigma}$.
\item We measure $\ket{h_\sigma}$ in the standard computational
basis and accept the input if the outcome of the last $l$ qubits
is the all-zero state. Thus, the probability of accepting
$\sigma$ is
$$Pr_{accept}(\sigma) = \frac{1}{t}\sum\limits_{i=1}^t\cos^2\frac{
  \pi k_i g_1(\sigma)}{m}\cdots \cos^2\frac{
  \pi k_i g_l(\sigma)}{m}.$$

If $f(\sigma)=1$ then all of $g_i(\sigma)=0$ and we will always
accept.

If $f(\sigma)=0$ then there is at least one such $j$ that
$g_j(\sigma)\neq0$ and the choice of the ``good'' set $K$
guarantees that the probability of the erroneously accepting is
bounded by
$$\begin{array}{rcl}
Pr_{accept}(\sigma) & = & \frac{1}{t}\sum\limits_{i=1}^t\cos^2\frac{
  \pi k_i g_1(\sigma)}{m}\cdots \cos^2\frac{\pi k_i g_l(\sigma)}{m}\\
& \leq & \frac{1}{t}\sum\limits_{i=1}^t\cos^2\frac{\pi k_i
g_j(\sigma)}{m} =
  \frac{1}{t}\sum\limits_{i=1}^t\frac{1}{2}\left(1+\cos\frac{
  2\pi k_i g_j(\sigma)}{m}\right) \\
& = & \frac{1}{2}+\frac{1}{2t}\sum\limits_{i=1}^t\cos\frac{2\pi k_i g_j(\sigma)}{m}\\
& \leq & \frac{1}{2} + \frac{\sqrt{\epsilon}}{2}.
\end{array}$$
\end{enumerate}

The number of qubits used by this QBP $Q$ is $qubits(Q)=O(\log\log m + |\chi_f|)$.
Therefore, the width of the program is
$2^{qubits(Q)}=O(2^{|\chi_f|}\log m)$.

\end{proof}

Note that though this upper bound is exponential our approach can be
effectively used when the size of a characteristic is
$O(\log\log{m})$ and $m=2^{n^{O(1)}}$. That is, in this case we will
stay within the polynomial width.

The Theorem \ref{LinearCharacteristicComputation} has two immediate
consequences which might be useful for proving upper bounds.
\begin{corollary}
If a Boolean function $f=f_1\wedge f_2\wedge\ldots \wedge f_s$ is a conjunction of
$s=O(\log n)$ Boolean functions $f_i$, each having a linear
characteristic polynomial over ${\mathbb Z}_{2^n}$, then $f$ can be
computed by an $O(\log{n})$-qubit quantum OBDD.
\end{corollary}
\begin{corollary}
If a Boolean function $f=f_1\vee f_2\vee\ldots\vee f_s$ is a
disjunction of $s=O(\log n)$ Boolean functions $f_i$ and the
negation of each has a linear characteristic polynomial over
${\mathbb Z}_{2^n}$, then $f$ can be computed by an
$O(\log{n})$-qubit quantum OBDD.
\end{corollary}

The last corollary uses the fact that $\neg f=\neg f_1\wedge\neg
f_2\wedge \ldots\wedge\neg f_s$.

The generalized approach can be used to construct an effective
quantum OBDD for the Boolean variant of the {\em Hidden Subgroup
Problem} \cite{ablayev-vasiliev:2009:EPTCS}.

\section{Using reductions for fingerprinting}

\emph{Reduction} is a well-known concept in complexity theory. The
most investigated reduction for the model of branching programs is
the technique of polynomial projections \cite{Wegener:2000:BP}.

\begin{definition}
The sequence $(f_n)$ of Boolean functions is a {polynomial
projection} of $(h_n)$, $(f_n)\leq_{\mbox{proj}}(h_n)$, if {$f_n =
h_{p(n)}(y_1, \ldots, y_{p(n)})$} for some polynomial $p$ and
$y_j\in\{0,1,x_1, \overline{x}_1, \ldots, x_n, \overline{x}_n\}$.
The number of $j$ such that $y_j\in\{x_i,\overline{x}_i\}$ is called
the multiplicity of $x_i$.
\end{definition}
%
%
%
%

Note that this type of reduction keep the linearity of the
corresponding characteristic polynomials and thus the effective
computability via the fingerprinting method.

\begin{lemma}
If {$(f_n)\leq_{\mbox{proj}}(h_n)$} and each $h_n$ can be represented by a linear
characteristic, then each $f_n$ also has a linear characteristic.
\end{lemma}

\begin{proof}
The variable substitutions corresponding to projection are $y_i=0$,
$y_i=1$, $y_i=x_{j_i}$, $y_i=1-x_{j_i}$. Obviously, if a
characteristic is linear then substitutions of these types to
polynomials $g=c_1y_1+\ldots c_ny_n+c_0$ keep them linear.
\end{proof}

On the other hand, the a special case of polynomial projections can
be used to prove lower bounds in the OBDD model.

\begin{definition}
A projection is called a read-once projection, $(f_n)\leq_{\mbox{rop}}(h_n)$,
if the multiplicity of each variable is bounded by one.
\end{definition}

The relation $\leq_{\mbox{rop}}$ is reflexive and transitive,
moreover if $(f_n)\leq_{\mbox{rop}}(h_n)$ and the OBDD size of $h_n$
is bounded by the polynomial $q(n)$ then the OBDD size of $f_n$ is
bounded by the polynomial $q(p(n))$. The last property was proved
for the classical OBDD model (see, e.g. \cite{Wegener:2000:BP}), but
it can be proved for the quantum setting in an analogous way.

Thus the read-once projections may be used for proving lower bounds
of the subclasses of Boolean functions whose projections are
exponentially hard for the OBDD model. These are, for example,
\emph{Set-Disjointness} and \emph{Neighbored Ones} Boolean functions
which were proved to be exponentially hard in
\cite{Sauerhoff-Sieling:2005:QBP}.

In particular these subclasses of ``hard'' functions cannot be
effectively computed via quantum fingerprinting method in the model
of read-once quantum branching programs.

Overall, the technique of polynomial projections outlines the limits of
our fingerprinting method in the quantum OBDD model with upper
bounds propagated by the general polynomial projections and lower
bounds based on individual ``hard'' functions and read-once
projections.

\bibliography{references}

\end{document}